\documentclass{article}

\usepackage{PRIMEarxiv}
\usepackage{amsthm}
\usepackage[utf8]{inputenc} 
\usepackage[T1]{fontenc} 
\usepackage{hyperref}      
\usepackage{url}         
\usepackage{booktabs}     
\usepackage{amsfonts}     
\usepackage{nicefrac}      
\usepackage{microtype}      
\usepackage{lipsum}
\usepackage{fancyhdr}      
\usepackage{graphicx}       
\usepackage{newtxtext,newtxmath}
\usepackage{amsthm}
\usepackage{latexsym}
\usepackage{bbm}
\usepackage{epsfig}
\usepackage{graphicx}
\usepackage{color}
\usepackage{ulem}
\usepackage{relsize}
\usepackage{xcolor}
\usepackage{float}
\usepackage{comment}
\usepackage{bm}
\usepackage{mathtools}
\usepackage{natbib} 
\setcitestyle{numbers} 
\newcommand{\cov}{\mathrm{cov}}
\newtheorem{proposition}{Proposition}
\newtheorem{lem}{Lemma}

\title{The Famous American Economist H. Markowitz and Mathematical Overview of his Portfolio Selection Theory
}

\author{
Ignas Gasparavičius, Andrius Grigutis \\
Institute of Mathematics\\
Faculty of Mathematics and Informatics \\
Vilnius University\\
Naugarduko g. 24, LT-03225, Vilnius, Lithuania\\
\\
\texttt{\{ignas.gasparavicius@mif.stud.vu.lt, andrius.grigutis@mif.vu.lt\}} \\
}

\begin{document}
\maketitle

\begin{abstract}
This survey article is dedicated to the life of the famous American economist H. Markowitz (1927--2023). We do revisit the main statements of the portfolio selection theory in terms of mathematical completeness including all the necessary auxiliary details.
\end{abstract}


\section{Introduction}

On June 22, 2023, the famous American economist H. Mar\-ko\-witz (aged 95) passed away. In this survey article, we revisit the main statements of the portfolio selection theory (also known as modern portfolio theory, mean-variance analysis, or just Markowitz theory) left by this prominent man, see \cite{P_Selection}, \cite{Efficient}, \cite{Merton}, \cite{CONSTANTINIDES}, \cite{50_1}, \cite{Math} and many other sources. We do not focus on the facts from Markowitz's life as this can be easily found in various other sources, see, for example, \cite{Part_I}, \cite{Part_II} or \cite{50}, but rather seek to reveal the mathematical completeness when formulating and proving statements in portfolio selection theory. We also briefly touch on the influence that Markowitz's theory made on other fields such as the rise of the capital asset pricing model, utility theory, arbitrage pricing theory, post-modernistic theory, etc.; see the end of the next section. 

To read the presented text, it is sufficient to be familiar with fundamental mathematical courses that are usually taught during the first two years in many universities of na\-tu\-ral sciences. We ask readers to be familiar with the Lagrange multiplier method, quadratic forms, differentials, linear algebra, probability theory, and mathematical statistics.    

Let $n\in\mathbb{N}$ and denote
\begin{align*}
\mathbb{W}:=(w_1,\,w_2,\,\ldots,\,w_n)\in\mathbb{R}^n, \, \sum_{i=1}^{n}w_i=1.
\end{align*}

We call the vector $\mathbb{W}\in\mathbb{R}^n$ {\bf portfolio}, whose every component $w_i,\,i=1,\,2,\,\ldots,\,n$ shall be understood as the proportion (weight) of asset $i$ in total "pool" of investments. In other words, if $A_1,\,A_2,\,\ldots,\,A_n$ denote certain financial firms, the portfolio informally can be represented as $"w_1A_1+w_2A_2+\ldots+w_nA_n"$, where $w_i$ denotes the proportion of capital allocated to the financial firm $A_i$, $i=1,\,2,\,\ldots,\,n$. If $w_i<0$ for some $i=1,\,2,\,\ldots,\,n$, we say the asset $i$ of the financial firm $A_i$ is in a "short position", which means that it is borrowed, sold, and the attained income invested in some other assets out of those represented by financial firms $A_1,\,A_2,\,\ldots,\,A_n$. Let ${\pmb R}:=(r_1,\,r_2,\,\ldots,\,r_n)$ be the collection of real-valued random variables defined on the same probability space. We suppose that the random vector ${\pmb R}$ represents the random returns "paid" by the financial firms $A_1,\,A_2,\,\ldots,\,A_n$. Of course, we quote the word "paid" here due to the negative return cases. Let ${\pmb \mu}:=(\mu_1,\,\mu_2,\,\ldots,\,\mu_n)$ and $\sigma^2_1,\,\sigma^2_2,\,\ldots,\,\sigma^2_n$ denote the expected values and variances of $r_1,\,r_2,\,\ldots,\,r_n$ correspondingly. In addition, let $\cov(r_i,\,r_j)$ denote the covariance between the random variables $r_i$ and $r_j$, for all $i,\,j=1,\,2,\,\ldots,\,n$ and
\begin{align*}
\Sigma:=
\begin{pmatrix}
\sigma^2_{1} & \cov(r_{1},\,r_{2})& \ldots & \cov(r_{1},\,r_{n})\\
\cov(r_{1},\,r_{2})&\sigma^2_{2}&\ldots&\cov(r_{2},\,r_{n})\\
\vdots & \vdots & \ddots & \vdots\\
\cov(r_{1},\,r_{n}) & \cov(r_{2},\,r_{n})& \ldots & \sigma^2_{n}
\end{pmatrix}
\end{align*}
be the covariance matrix. In practice, the collections of numerical characteristics ${\pmb \mu}$ and $\Sigma$ can be obtained variously: according to some historical observations, estimates, predictions, etc, see \cite{covar}, \cite{covar_1}, \cite{covar_2}. We next denote the random variable $P:=w_1r_1+w_2r_2+\ldots+w_nr_n$ whose expected value and variance are:
\begin{align*}
&\mu_p:={\pmb \mu}\mathbb{W}^T,\\
&\sigma^2_P:=\mathbb{W}\Sigma\mathbb{W}^T=\sum_{i=1}^{n}w^2_i\sigma^2_i+2\sum_{1\leqslant i<j\leqslant n}w_iw_j\cov(r_i,\,r_j).
\end{align*}

The notation $P$ here is of course associated with the portfolio as it represents the random return of sum of weighted random returns $r_1,\,r_2,\,\ldots,\,r_n$ of single investments.

The famous portfolio selection theory asks what weights $\mathbb{W}\in\mathbb{R}^n$ shall be chosen so that $\sigma_P^2$ and $\mu_P$ satisfy certain conditions (see the next section for more precise formulations). The core assumption of portfolio selection theory is that investors behave rationally: they prefer the lower-risk (variance) portfolios under the same level of the expected return. Technically, the numerical characteristics $\sigma_P^2$ and $\mu_P$ are nothing but functions of $n$ variables $w_1,\,w_2,\,\ldots,\,w_n$ under constraint $\mathbbm{1}\mathbb{W}^T=1$, where $\mathbbm{1}:=(1,\,1,\ldots,\,1)_{1\times n}$.

In the next section, we reformulate the core statements of the Markowitz portfolio selection theory. As we will see, the portfolio selection is based on the inverse $\Sigma^{-1}$ of the covariance matrix $\Sigma$. Because of that, we need the following two little auxiliary statements on the existence of $\Sigma^{-1}$ and some of its properties. 

\begin{lem}\label{lem:nonsingular}
Let ${\pmb 0}:=(0,\,0,\ldots,\,0)_{1\times n}$. The covariance matrix $\Sigma$ is singular iff $\mathbb{P}(w_1r_1+w_2r_2+\ldots+w_nr_n=conts.)=1$ for some $\mathbb{W}\in\mathbb{R}^n\setminus\{{\pmb 0}\}$.
\end{lem}

\begin{proof}
{\textsc I}f $\Sigma$ is singular then 
\begin{align}\label{singular}
\Sigma {\mathbb{W}}^T={\pmb 0}^T
\end{align}
for some $\mathbb{W}\in\mathbb{R}^n\setminus\{{\pmb 0}\}$. Consequently, eq. \eqref{singular} implies 
\begin{align*}
\sigma_P^2=\mathbb{W}\Sigma {\mathbb{W}}^T=\mathbb{W}{\pmb 0}^T=0,
\end{align*}
and that means $\mathbb{P}(w_1r_1+w_2r_2+\ldots+w_nr_n=conts.)=1$ since the variance of the random variable $w_1r_1+w_2r_2+\ldots+w_nr_n$ is zero. 

    Conversely, $\mathbb{P}(w_1r_1+w_2r_2+\ldots+w_nr_n=conts.)=1$ implies 
\begin{align*}
\sigma^2_P=\mathbb{W}\Sigma\mathbb{W}^T=0=\mathbb{W}{\pmb 0}^T\quad \Rightarrow \quad \Sigma\mathbb{W}^T={\pmb 0}^T,
\end{align*}
which implies the singularity of $\Sigma$.
\end{proof}

\bigskip

Notice that the condition $\mathbb{P}(w_1r_1+w_2r_2+\ldots+w_nr_n=conts.)=1$ in Lemma \ref{lem:nonsingular} is equivalent to the condition that all the realizations $(x_1,\,x_2,\,\ldots,\,x_n)\in\mathbb{R}^n$ of the random vector ${\pmb R}$ lie on the hyperplane $w_1x_1+w_2x_2+\ldots+w_nx_n=conts.$ almost surely. See \cite{Cov} as a related topic of mean-variance analysis under the singular covariance matrix.

\begin{lem}\label{lem:symmetric_matrix}
Let $A$ be the symmetric and positively defined matrix such that its inverse $A^{-1}$ exists. Then $A^{-1}$ is also symmetric and positively defined. 
\end{lem}
\begin{proof}
{\textsc L}et $I$ be the identity matrix. Then $AA^{-1}=A^{-1}A=I$ and consequently
\begin{align*}
\left(AA^{-1}\right)^T=I^T\Rightarrow
\left(A^{-1}\right)^TA^T=I\Rightarrow
\left(A^{-1}\right)^TA=I\Rightarrow
\left(A^{-1}\right)^T=A^{-1}.
\end{align*}
If $A$ is positively defined matrix, then $\mathbb{W}A\mathbb{W}^T>0$ with any $\mathbb{W}\in\mathbb{R}^n\setminus\{{\pmb 0}\}$. If $\mathbb{W}=\tilde{\mathbb{W}}A^{-1}$ then
\begin{align*}
\mathbb{W}A\mathbb{W}^T=\tilde{\mathbb{W}}A^{-1}AA^{-1}\tilde{\mathbb{W}}^T=\tilde{\mathbb{W}}A^{-1}\tilde{\mathbb{W}}^T>0
\end{align*}
with any $\tilde{\mathbb{W}}=\mathbb{W}A\in\mathbb{R}^n\setminus\{{\pmb 0}\}$.
\end{proof}

Let us point to the reference \cite{Strang_2023} as a great source on linear algebra.

\section{Propositions and their implications}\label{sec:Propop}

As stated, in this section we formulate the essential statements of portfolio selection theory. These statements allow us to choose portfolios with the lowest variance, and the highest Sharpe ratio, draw the efficient frontier, and illustrate the fund separation theorem.

\begin{proposition}[The minimal variance portfolio]\label{prop:var}
Let $P={\pmb R}\mathbb{W}^T$ be the random variable, where ${\pmb R}$ denotes the random vector whose covariance matrix $\Sigma$ is non-singular, and $\mathbb{W}\in\mathbb{R}^n$. The variance $\sigma_P^2=\mathbb{W}\Sigma\mathbb{W}^T$, under constraint $\mathbbm{1}\mathbb{W}^T=1$, is minimal when
\begin{align}\label{min_var_w}
\mathbb{W}^T=\frac{\Sigma^{-1}\mathbbm{1}^T}{\mathbbm{1}\Sigma^{-1}\mathbbm{1}^T}.
\end{align}

Moreover,
\begin{align*}
\min\sigma_{P}^2=\frac{1}{\mathbbm{1}\Sigma^{-1}\mathbbm{1}^T}.
\end{align*}
\end{proposition}

The next statement provides what weighs $\mathbb{W}\in\mathbb{R}^n$ maximize the Sharpe ratio. First, let us denote $\tilde{\pmb \mu}:=(\mu_1-r_f,\,\mu_2-r_f,\,\ldots,\,\mu_n-r_f)$, where $r_f=conts.$ denotes the risk-free return rate.

\begin{proposition}[The maximal Sharpe ratio portfolio]\label{prop:Sharpe}
Let $P={\pmb R}\mathbb{W}^T$ be the random variable, where ${\pmb R}$ denotes the random vector whose covariance matrix $\Sigma$ is non-singular, $\tilde{\pmb \mu}\neq{\pmb 0}$, and $\mathbb{W}\in\mathbb{R}^n$. The Sharpe ratio
\begin{align*}
S_P(\mathbb{W})=\frac{\mathbb{W}\tilde{\boldsymbol{\mu}}^T}{\sqrt{\mathbb{W}\Sigma\mathbb{W}^T}},
\end{align*}
under constraint $\mathbbm{1}\mathbb{W}^T=1$, is maximal when
\begin{align}\label{max_S_portfolio}
\mathbb{W}^T=\frac{\Sigma^{-1}\tilde{\boldsymbol{\mu}}^T}{\mathbbm{1}\Sigma^{-1}\tilde{\boldsymbol{\mu}}^T}.
\end{align}
Moreover,
\begin{align}\label{max_Sharpe}
\max S_P=\sqrt{\tilde{\boldsymbol{\mu}}\Sigma^{-1}\tilde{\boldsymbol{\mu}}^T}\cdot\frac{\left|\mathbbm{1}\Sigma^{-1}\tilde{\boldsymbol{\mu}}^T\right|}{\mathbbm{1}\Sigma^{-1}\tilde{\boldsymbol{\mu}}^T}.
\end{align}
\end{proposition}

\begin{proposition}[The minimal variance portfolio under the given return]\label{prop:min_var_mu_const}
Let $P={\pmb R}\mathbb{W}^T$ be the random variable, where ${\pmb R}$ denotes the random vector whose covariance matrix $\Sigma$ is non-singular, ${\pmb \mu}\neq{\pmb 0}$, and $\mathbb{W}\in\mathbb{R}^n$. The variance $\sigma_P^2=\mathbb{W}\Sigma\mathbb{W}^T$, under constraints $\mathbbm{1}\mathbb{W}^T=1$ and ${\pmb \mu}\mathbb{W}^T=\mu_0=$conts., is minimal when
\begin{align}\label{weights_ef_no_rf}
\mathbb{W}^T=\left(\Sigma^{-1}{\pmb\mu}^T,\,\Sigma^{-1}\mathbbm{1}^T\right)
\begin{pmatrix}
{\pmb\mu}\Sigma^{-1}{\pmb\mu}^T&{\pmb\mu}\Sigma^{-1}\mathbbm{1}^T\\
\mathbbm{1}\Sigma^{-1}{\pmb\mu}^T&\mathbbm{1}\Sigma^{-1}\mathbbm{1}^T
\end{pmatrix}^{-1}
\begin{pmatrix}
\mu_{0}\\
1
\end{pmatrix}.
\end{align}
\end{proposition}

Proposition \ref{prop:min_var_mu_const} plays the key role in drawing the curve in $(\sigma,\,\mu)$ plane which is called the {\bf efficient frontier}. The efficient frontier denotes the set of portfolios where no other portfolio exists with a higher expected return $\mu_P$ under the same standard deviation $\sigma_P$ (risk) of return. Let us denote $\mu_{\sigma \min}$ the expected return of portfolio \eqref{min_var_w}, i.e.
\begin{align*}
\mu_{\sigma \min}:=\frac{{\pmb\mu} \Sigma^{-1}\mathbbm{1}^T}{\mathbbm{1}\Sigma^{-1}\mathbbm{1}^T}.
\end{align*}

\begin{proposition}[Efficient frontier of risky assets only]\label{prop:EF_risk}
An efficient frontier when the portfolio consists of risky assets only is
\begin{align}\label{eq:EF_risky_only}
\sigma^2=a\mu^2+b\mu+c,\,\mu\geqslant \mu_{\min\sigma},
\end{align}
where
\begin{align*}
&a=\frac{\mathbbm{1}\Sigma^{-1}\mathbbm{1}^T}{d},\,
b=-2\cdot\frac{{\pmb\mu}\Sigma^{-1}\mathbbm{1}^T}{d},\,
c=\frac{{\pmb\mu}\Sigma^{-1}{\pmb\mu}^T}{d},\\
&d=\left({\pmb\mu}\Sigma^{-1}{\pmb\mu}^T\right)\left(\mathbbm{1}\Sigma^{-1}\mathbbm{1}^T\right)-\left({\pmb\mu}\Sigma^{-1}\mathbbm{1}^T\right)^2>0 \text{ if } {\pmb \mu}\neq{\pmb 0}.
\end{align*}

\end{proposition}

Notice the condition $\mu\geqslant\mu_{\min \sigma}$ in \eqref{eq:EF_risky_only} defines the upper branch of the hyperbola in the plane $(\mu,\,\sigma)$. By letting $\mu\in\mathbb{R}$ in \eqref{eq:EF_risky_only}, we get the total hyperbola which is often referred to as the {\bf minimum variance frontier}. We call the portfolios, lying on the efficient frontier, {\bf efficient}. According to the core assumption of the Markowitz theory, efficient portfolios are always preferred by rational investors.  

\begin{proposition}[Efficient frontier including the risk-free asset]\label{prop:EF_with_rf}
An efficient frontier when the portfolio consists of risky assets and the risk-free asset is
\begin{align*}
\mu=\sqrt{\tilde{\pmb{\mu}}\Sigma^{-1}\tilde{\pmb{\mu}}^T}\cdot\sigma+r_f,\,\sigma \geqslant 0.
\end{align*}
\end{proposition}

We call the portfolio \eqref{max_S_portfolio} {\bf optimal} and denote it by $M$. In addition, by $\sigma_M$ and $\mu_M$ we denote the standard deviation and the expected return of the optimal portfolio correspondingly.

\begin{proposition}[Tangent line of the efficient frontier at optimal portfolio]\label{prop:Tangent}
The tangent line of the efficient frontier given in \eqref{eq:EF_risky_only} at the optimal portfolio $M$ is
\begin{align*}
\mu=\frac{\mu_M-r_f}{\sigma_M}\cdot\sigma+r_f,\,\sigma\geqslant0.
\end{align*}
\end{proposition}

Notice that the slope of the semi-line in Proposition \ref{prop:Tangent} is the Sharpe ratio of the optimal portfolio. Thus, $(\mu_M-r_f)/\sigma_M$ equals to the maximal Sharpe ratio given in \eqref{max_Sharpe}. Moreover, the slope of the semi-line in Proposition \ref{prop:EF_with_rf} equals the maximal Sharpe ratio too when $\mu_M-r_f>0$. In addition, the semi-line in Proposition \ref{prop:Tangent} can be easily rewritten in its parametric form:
\begin{align*}
\begin{cases}
\mu=w\mu_M+(1-w)r_f\\
\sigma=w\sigma_M
\end{cases},\,w\geqslant0.
\end{align*}
Thus, portfolios of the efficient frontier, when there is risk-free investment included, are just the combinations of optimal portfolio $M$ and risk-free asset $F$, i.e. "$wM+(1-w)F,\,w\geqslant0$".

\begin{proposition}[Mutual fund separation theorem]\label{prop:Separation}
Let $\mathbb{W}_1,\,\mathbb{W}_2,\,\ldots,\,\mathbb{W}_m$, where $m\in\mathbb{N}$ and $\mathbb{W}_i\in\mathbb{R}^n$ for all $i=1,\,2,\,\ldots,\,m$, be the efficient portfolios whose expected returns are $\mu_{0,\,1},\,\mu_{0,\,2},\,\ldots,\,\mu_{0,\,m}$ respectively. Let
$\mathbb{W}\in\mathbb{R}^m$ be such that
\begin{align}\label{syst}
\begin{cases}
w_1+w_2+\ldots+w_m=1\\
w_1\mu_{0,\,1}+w_2\mu_{0,\,2}+\ldots+w_m\mu_{0,\,m}=\tilde{\mu}_0\geqslant\mu_{\sigma \min}
\end{cases}.
\end{align}
Then the portfolio $w_1\mathbb{W}_1+w_2\mathbb{W}_2+\ldots+w_m\mathbb{W}_m$ is also efficient.
\end{proposition}

If $m=2$ in Proposition \ref{prop:Separation} and $\mu_{0,\,1}\neq\mu_{0,\,2}$, the system \eqref{syst} implies
\begin{align}\label{two-fund}
(w_1,\,w_2)=\left(\frac{\tilde{\mu}_0-\mu_{0,\,2}}{\mu_{0,\,1}-\mu_{0,\,2}},\,\frac{\mu_{0,\,1}-\tilde{\mu}_0}{\mu_{0,\,1}-\mu_{0,\,2}}\right).
\end{align}
These proportions \eqref{two-fund} are met in the so-called two-fund separation theorem. Notice that the second equation in \eqref{syst} is nothing but the expected return of the portfolio $w_1\mathbb{W}_1+w_2\mathbb{W}_2+\ldots+w_m\mathbb{W}_m$. In practice, the mutual fund separation theorem is valuable because under some favorable circumstances it is more convenient or rational, in terms of transaction costs, to acquire the portfolios (funds) $\mathbb{W}_1,\,\mathbb{W}_2,\,\ldots,\,\mathbb{W}_m$ and form the desired efficient portfolio out of them rather than to acquire the single assets of the financial firms $A_1,\,A_2,\,\ldots,\,A_n$; see \cite{Fund}, \cite{Fund_1}, \cite{Fund_2}.

The listed Propositions \ref{prop:var}--\ref{prop:Separation} consist of the essence of the portfolio selection theory. As mentioned, Markowit's works had a significant influence on many other fields or views of the surrounding phenomena. For instance, the shape of the efficient frontier tells us that ''extra benefits are not free''. In other words, the higher expected return corresponds to the higher standard deviation (risk), and the level of risk tolerance is up to the investors themselves. That encouraged developments in {\bf risk aversion} and {\bf utility theory}, see \cite{utility}, \cite{utility_1}. Also, the understanding of risk was extended realizing that the standard deviation as such does not tell the whole story. As a measure of returns' scattering around its expectation, the standard deviation and similar measures provide only the so-called own (idiosyncratic) risk of the asset that generates the return. The {\bf capital asset pricing model} (see \cite{CAPM}, \cite{CAPM_2}, \cite{CAPM_1}, \cite{CAPM_3}, \cite[p. 132]{Muller_1988}, \cite{CAPM_4}) states that under some ideal-word assumptions, the expected return of any asset $A$ is implied by the equation
\begin{align}\label{CAPM}
\mathbb{E}r_A=r_f+\beta_i(\mathbb{E}r_S-r_f),\,\beta_i=\frac{\cov(r_A,\,r_S)}{\sigma^2_S},
\end{align}
where $S$ in \eqref{CAPM} denotes the systemic (market) portfolio. The coefficient $\beta_i$ in \eqref{CAPM}, the slope of this linear equation, is interpreted as a measure of systemic risk that is raised by the systemic portfolio $S$. Based on \eqref{CAPM}, in $(\beta,\,\mu)$ plane one may define the so-called {\bf security market line}   
\begin{align}\label{SML}
\mu=r_f+\beta(\mathbb{E}r_S-r_f),\,\beta\in\mathbb{R}.
\end{align}
The line \eqref{SML} is a basic tool providing the equilibrium of the surrounding returns, see \cite{SML}, \cite{SML_1}. In \cite[p. 65]{KNU} one may find a certain related instance where the foundations of the capital asset pricing model appear. It is said that a certain asset is fairly evaluated if its expected return belongs to the line \eqref{SML}. On the other hand, assets below the line \eqref{SML} are known as overestimated, while the above ones are called underestimated. 
Based on that, there was built the {\bf arbitrage pricing theory}, see \cite{arbitrage}, \cite{arbitrage_1}. The arbitrage pricing theory essentially states that under certain assumptions the arbitrage (guaranteed profit with zero investments and no risk) is possible if there are assets whose returns do not lie on the line \eqref{SML}. Last but not least, the vector of expected returns and covariance matrix are just two out of plenty of other numerical characteristics. The ideas of the mean-variance analysis inspired the development of the so-called {\bf post-modernistic theory}. The post-modernistic theory optimizes plenty of other numerical cha\-rac\-te\-ris\-tics (quantiles, etc.) in a similar way that Markowitz's theory does the job in mean-variance optimization, see \cite{post}, \cite{post_1}, \cite{post_2}, \cite{post_3}, \cite{post_4}, \cite{post_5}. 

\section{Proofs}

\begin{proof}[Proof of Proposition \ref{prop:var}.]
We shall find the point $\mathbb{W}\in\mathbb{R}^n$ such that the function $\sigma_P^2=\mathbb{W}\Sigma\mathbb{W}^T$ attains its minimum under constraint $\mathbbm{1}\mathbb{W}^T=1$. Let us setup the Lagrangian function $L(\mathbb{W},\,\lambda)=\mathbb{W}\Sigma\mathbb{W}^T-\lambda(\mathbbm{1}\mathbb{W}^T-1)$. By calculating the partial derivatives of $L(\mathbb{W},\,\lambda)$ with respect to $w_1,\,w_2,\,\ldots,\,w_n$ and $\lambda$, we obtain the following system of linear equations
\begin{align}\label{syst:_min_var}
\begin{cases}
2\Sigma\mathbb{W}^T-\lambda \mathbbm{1}^T={\pmb 0}^T,\\
\mathbbm{1}\mathbb{W}^T=1.
\end{cases}
\end{align}
From the first equation in \eqref{syst:_min_var} we get
\begin{align*}
\mathbb{W}^T=\frac{\lambda}{2}\Sigma^{-1}\mathbbm{1}^T.
\end{align*}
This and the last equation in \eqref{syst:_min_var} yields
\begin{align*}
\mathbbm{1}\frac{\lambda}{2}\Sigma^{-1}\mathbbm{1}^T=1 \quad \Rightarrow \quad \frac{\lambda}{2}=\frac{1}{\mathbbm{1}\Sigma^{-1}\mathbbm{1}^T}.
\end{align*}
Thus,
\begin{align}\label{min_sigma}
\mathbb{W}^T=\frac{\Sigma^{-1}\mathbbm{1}^T}{\mathbbm{1}\Sigma^{-1}\mathbbm{1}^T}.
\end{align}

To prove that the point \eqref{min_sigma} is minimum, we observe that the Hessian matrix 
\begin{align*}
\frac{\partial^2 L(\mathbb{W},\,\lambda)}{\partial w_i \partial w_j}=\frac{\partial \left(2\Sigma\mathbb{W}^T-\lambda \mathbbm{1}^T\right)}{\partial w_i}=2\Sigma,
\end{align*}
for all $i,\,j=1,\,2,\,\ldots,\,n$. This yields the positive second differential:
\begin{align*}
d^2\sigma_P^2(\mathbb{W})=(\Delta w_1,\,\Delta w_2,\,\ldots,\Delta w_n)2\Sigma(\Delta w_1,\,\Delta w_2,\,\ldots,\Delta w_n)^T>0,
\end{align*}
for all $(\Delta w_1,\,\Delta w_2,\,\ldots,\,\Delta w_n)\in\mathbb{R}^n\setminus\{\pmb 0\}$. 

By inserting the point \eqref{min_sigma} into $\sigma_P^2(\mathbb{W})=\mathbb{W}\Sigma\mathbb{W}^T$ we get
\begin{align*}
\min \sigma^2=\frac{\mathbbm{1}\Sigma^{-1}}{\mathbbm{1}\Sigma^{-1}\mathbbm{1}^T}
\Sigma
\frac{\Sigma^{-1}\mathbbm{1}^T}{\mathbbm{1}\Sigma^{-1}\mathbbm{1}^T}
=\frac{1}{\mathbbm{1}\Sigma^{-1}\mathbbm{1}^T}.
\end{align*}
Here we have used the symmetry of $\Sigma^{-1}$, see Lemma \ref{lem:symmetric_matrix}. 
\end{proof}

\bigskip

\begin{proof}[Proof of Proposition \ref{prop:Sharpe}.]
Let us set up the Lagrangian function
\begin{align*}
L(\mathbb{W},\lambda)=\frac{\mathbb{W}\boldsymbol{\mu}^T-r_f}{\sqrt{\mathbb{W}\Sigma\mathbb{W}^T}}-\lambda(\mathbbm{1}\mathbb{W}^T-1).
\end{align*}
Its partial derivatives with respect to $w_1,\,w_2,\,\ldots,\,w_n$ are
\begin{align*}
\frac{\partial L(\mathbb{W},\,\lambda)}{\partial w_i}=\frac{{\pmb \mu}^T\sqrt{\mathbb{W}\Sigma\mathbb{W}^T}-(\mathbb{W}{\pmb \mu}^T-r_f)\frac{\Sigma\mathbb{W}^T}{\sqrt{\mathbb{W}\Sigma\mathbb{W}^T}}}{\mathbb{W}\Sigma\mathbb{W}^T}
-\lambda\mathbbm{1}^T,
\end{align*}
for all $i=1,\,2,\,\ldots,\,n$. This and the partial derivative of $L(\mathbb{W},\,\lambda)$ with respect to $\lambda$, yields the following system of linear equations
\begin{align}\label{syst:max_S}
\begin{cases}
{\pmb \mu}^T\left(\mathbb{W}\Sigma\mathbb{W}^T\right)-\left(\mathbb{W}{\pmb \mu}^T-r_f\right)\Sigma\mathbb{W}^T
-\lambda\mathbbm{1}^T\left(\mathbb{W}\Sigma\mathbb{W}^T\right)^{3/2}={\pmb 0}^T,\\
\mathbbm{1}\mathbb{W}^T=1.
\end{cases}
\end{align}
By multiplying the both sides of the first equation in \eqref{syst:max_S} by $\mathbb{W}$ and simplifying, we obtain
\begin{align*}
\mathbb{W}{\pmb \mu}^T-(\mathbb{W}{\pmb \mu}^T-r_f)-\lambda\mathbb{W}\mathbbm{1}^T\left(\mathbb{W}\Sigma\mathbb{W}^T\right)^{1/2}=0
\quad \Rightarrow \quad
\lambda=\frac{r_f}{\left(\mathbb{W}\Sigma\mathbb{W}^T\right)^{1/2}}.
\end{align*}
By inserting this expression of $\lambda$ into the first equation of \eqref{syst:max_S}, we obtain
\begin{align}\label{eq:before_Hess}
&{\pmb \mu}^T\left(\mathbb{W}\Sigma\mathbb{W}^T\right)-\left(\mathbb{W}{\pmb\mu}^T-r_f\right)\Sigma\mathbb{W}^T
-r_f\mathbbm{1}^T\left(\mathbb{W}\Sigma\mathbb{W}^T\right)={\pmb 0}^T\\ \nonumber
&\Rightarrow\,
\left(\mathbb{W}{\pmb \mu}^T-r_f\right)\Sigma\mathbb{W}^T=\left(\mathbb{W}\Sigma\mathbb{W}^T\right)({\pmb \mu}^T-r_f\mathbbm{1}^T)\\ \nonumber
&\Rightarrow\,
\mathbb{W}^T=\frac{\mathbb{W}\Sigma\mathbb{W}^T}{\mathbb{W}{\pmb \mu}^T-r_f}\Sigma^{-1}({\pmb \mu}-r_f\mathbbm{1})^T,\,\tilde{\pmb\mu}\neq{\pmb 0}.
\end{align}
By inserting this into the last equation of \eqref{syst:max_S}, under the requirement $\tilde{\pmb\mu}\neq{\pmb 0}$, we get
\begin{align*}
\mathbbm{1}\frac{\mathbb{W}\Sigma\mathbb{W}^T}{\mathbb{W}{\pmb \mu}^T-r_f}\Sigma^{-1}({\pmb \mu}-r_f\mathbbm{1})^T=1
\quad \Rightarrow \quad
\frac{\mathbb{W}\Sigma\mathbb{W}^T}{\mathbb{W}{\pmb \mu}^T-r_f}=\frac{1}{\mathbbm{1}\Sigma^{-1}({\pmb \mu}-r_f\mathbbm{1})^T}.
\end{align*}
Thus,
\begin{align}\label{max_S}
\mathbb{W}^T=\frac{\Sigma^{-1}({\pmb \mu}-r_f\mathbbm{1})^T}{\mathbbm{1}\Sigma^{-1}({\pmb \mu}-r_f\mathbbm{1})^T}
=\frac{\Sigma^{-1}\tilde{{\pmb \mu}}^T}{\mathbbm{1}\Sigma^{-1}\tilde{{\pmb \mu}}^T}.
\end{align}

To prove that the point \eqref{max_S} is maximum, we first compute the Sharpe ratio in this point: 
\begin{align*}
S_P(\mathbb{W})&=\frac{\mathbb{W}\tilde{{\pmb \mu}}^T}{\sqrt{\mathbb{W}\Sigma\mathbb{W}^T}}
=\frac{\tilde{{\pmb\mu}}\Sigma^{-1}\tilde{{\pmb \mu}}^T}{\mathbbm{1}\Sigma^{-1}\tilde{{\pmb \mu}}^T}\Bigg{/}
\sqrt{\frac{\tilde{{\pmb\mu}}\Sigma^{-1}}{\mathbbm{1}\Sigma^{-1}\tilde{\pmb{\mu}}^T}\cdot \Sigma \cdot\frac{\Sigma^{-1}\tilde{{\pmb \mu}}^T}{\mathbbm{1}\Sigma^{-1}\tilde{{\pmb \mu}}^T}}
=\sqrt{\tilde{{\pmb\mu}}\Sigma^{-1}\tilde{{\pmb \mu}}^T}\cdot
\frac{\left|\mathbbm{1}\Sigma^{-1}\tilde{{\pmb\mu}}^T \right|}{\mathbbm{1}\Sigma^{-1}\tilde{{\pmb \mu}}^T}.
\end{align*}
Let us take some other point in $\mathbb{R}^n$ than \eqref{max_S}
\begin{align}\label{max_S_delta}
\tilde{\mathbb{W}}^T=\frac{\Sigma^{-1}\tilde{{\pmb\mu}}^T}{\mathbbm{1}\Sigma^{-1}\tilde{{\pmb\mu}}^T}+{\pmb \Delta}^T,
\end{align}
where ${\pmb \Delta}:=(\delta_1,\,\ldots,\,\delta_n)\in\mathbb{R}^n$. Then the Sharpe ratio in \eqref{max_S_delta} is
\begin{align*}
S_P(\tilde{\mathbb{W}})=
\left(\frac{\tilde{\pmb\mu}\Sigma^{-1}\tilde{\pmb\mu}^T}{\mathbbm{1}\Sigma^{-1}\tilde{\pmb\mu}^T}+{\pmb \Delta}\tilde{\pmb \mu}^T\right)
\Bigg{/}
\left(\sqrt{\frac{\tilde{\pmb \mu}\Sigma^{-1}\tilde{\pmb \mu}^T}{\left(\mathbbm{1}\Sigma^{-1}\tilde{\pmb \mu}^T\right)^2}+\frac{2\tilde{\pmb \mu}{\pmb \Delta}^T}{\mathbbm{1}\Sigma^{-1}\tilde{\pmb \mu}^T}+{\pmb \Delta}\Sigma{\pmb \Delta}^T}\right).
\end{align*}
Thus, the desired statement follows if we can show
\begin{align}\label{nelyg}
S_P(\tilde{\mathbb{W}})<S_P({\mathbb{W}})
\end{align}
for all ${\pmb \Delta}=(\delta_1,\,\ldots,\,\delta_n)\in\mathbb{R}^n\setminus\{{\pmb 0}\}$. 

Inequality \eqref{nelyg} is equivalent to
\begin{align}\label{desired_ineq}
\frac{\tilde{\pmb \mu}\Sigma^{-1}\tilde{\pmb \mu}^T}{\mathbbm{1}\Sigma^{-1}\tilde{\pmb \mu}^T}+{\pmb \Delta}\tilde{\pmb \mu}^T<
\sqrt{\left(\frac{\tilde{\pmb \mu}\Sigma^{-1}\tilde{\pmb \mu}^T}{\mathbbm{1}\Sigma^{-1}\tilde{\pmb \mu}^T}+{\pmb \Delta}\tilde{\pmb \mu}^T\right)^2
-\left({\pmb \Delta}\tilde{\pmb \mu}^T\right)^2+{\pmb \Delta}\Sigma{\pmb \Delta}^T\cdot\tilde{\pmb\mu}\Sigma^{-1}\tilde{\pmb\mu}^T}.
\end{align}
Let us observe that the inequality 
\begin{align*}
x<\sqrt{x^2+\varepsilon},\,x\in\mathbb{R}\quad \Leftrightarrow \quad \frac{x}{|x|}<\sqrt{1+\frac{\varepsilon}{x^2}},\,x\in\mathbb{R}\setminus\{0\}
\end{align*}
is correct for all $\varepsilon>0$. Thus, inequality \eqref{desired_ineq} follows if we can show 
\begin{align}\label{desired_ineq_2}
\left({\pmb \Delta}\tilde{\pmb \mu}^T\right)^2<{\pmb \Delta}\Sigma{\pmb \Delta}^T\cdot\tilde{\pmb \mu}\Sigma^{-1}\tilde{\pmb\mu}^T.
\end{align}

Due to the positively defined covariance matrix, we can write
\begin{align*}
\left(\mathbbm{1}\Sigma^{-1}\tilde{\pmb \mu}^T\right)^2\tilde{\mathbb{W}}\Sigma\tilde{\mathbb{W}}^T
={\pmb \Delta}\Sigma{\pmb \Delta}^T\left(\mathbbm{1}\Sigma^{-1}\tilde{\pmb\mu}^T\right)^2+2{\pmb \Delta}\tilde{\pmb \mu}^T\left(\mathbbm{1}\Sigma^{-1}\tilde{\pmb\mu}^T\right)+\tilde{\pmb\mu}\Sigma^{-1}\tilde{\pmb\mu}^T>0.
\end{align*}
This implies \eqref{desired_ineq_2} because $ax^2+bx+c>0$ being valid for all $x\in\mathbb{R}$ implies $b^2-4ac<0$.


\end{proof}

\begin{proof}[Proof of Proposition \ref{prop:min_var_mu_const}.]
Let us set up the Lagrangian function
\begin{align*}
L(\mathbb{W},\,\lambda_1,\,\lambda_2)=\mathbb{W}\Sigma\mathbb{W}^T-\lambda_1\left({\pmb \mu}\mathbb{W}^T-\mu_0\right)-\lambda_2\left(\mathbbm{1}\mathbb{W}^T-1\right).
\end{align*}
By calculating its partial derivatives with respect to $w_1,\,w_2,\,\ldots,\,w_n$, $\lambda_1$ and $\lambda_2$, we get
\begin{align}\label{syst:ef_n}
\begin{cases}
&2\Sigma\mathbb{W}^T-\lambda_1{\pmb \mu}^T-\lambda_2\mathbbm{1}^T={\pmb 0} \\
&{\pmb \mu}\mathbb{W}^T=\mu_0 \\
&\mathbbm{1}\mathbb{W}^T=1
\end{cases}.
\end{align}
The first equation of \eqref{syst:ef_n} yields
\begin{align}\label{lygt:ef_n}
\mathbb{W}^T=\frac{\lambda_1}{2}\Sigma^{-1}{\pmb \mu}^T+\frac{\lambda_2}{2}\Sigma^{-1}\mathbbm{1}^T
=\left(\Sigma^{-1}{\pmb \mu}^T,\,\Sigma^{-1}\mathbbm{1}^T\right)
\begin{pmatrix}
\lambda_1/2\\\lambda_2/2
\end{pmatrix}.
\end{align}

We then multiply the both sides of \eqref{lygt:ef_n} by ${\pmb\mu}$ and $\mathbbm{1}$ respectively and, due to the last two equations in \eqref{syst:ef_n}, obtain
\begin{align*}
\begin{cases}
\frac{\lambda_1}{2}{\pmb \mu}\Sigma^{-1}{\pmb \mu}^T+\frac{\lambda_2}{2}{\pmb \mu}\Sigma^{-1}\mathbbm{1}^T=\mu_0\\
\frac{\lambda_1}{2}\mathbbm{1}\Sigma^{-1}{\pmb\mu}^T+\frac{\lambda_2}{2}\mathbbm{1}\Sigma^{-1}\mathbbm{1}^T=1
\end{cases}
\Leftrightarrow
\begin{pmatrix}
{\pmb\mu}\Sigma^{-1}{\pmb\mu}^T&{\pmb\mu}\Sigma^{-1}\mathbbm{1}^T\\
\mathbbm{1}\Sigma^{-1}{\pmb\mu}^T&\mathbbm{1}\Sigma^{-1}\mathbbm{1}^T
\end{pmatrix}
\begin{pmatrix}
\frac{\lambda_1}{2}\\\frac{\lambda_2}{2}
\end{pmatrix}
=
\begin{pmatrix}
\mu_0\\
1
\end{pmatrix}.
\end{align*}

Thus,
\begin{align*}
\begin{pmatrix}
\lambda_1/2\\
\lambda_2/2
\end{pmatrix}
=
\begin{pmatrix}
{\pmb\mu}\Sigma^{-1}{\pmb\mu}^T&{\pmb\mu}\Sigma^{-1}\mathbbm{1}^T\\
\mathbbm{1}\Sigma^{-1}{\pmb\mu}^T&\mathbbm{1}\Sigma^{-1}\mathbbm{1}^T
\end{pmatrix}^{-1}
\begin{pmatrix}
\mu_{0}\\
1
\end{pmatrix}
\end{align*}
and
\begin{align*}
\mathbb{W}^T=\left(\Sigma^{-1}{\pmb\mu}^T,\,\Sigma^{-1}\mathbbm{1}^T\right)
\begin{pmatrix}
{\pmb\mu}\Sigma^{-1}{\pmb\mu}^T&{\pmb\mu}\Sigma^{-1}\mathbbm{1}^T\\
\mathbbm{1}\Sigma^{-1}{\pmb\mu}^T&\mathbbm{1}\Sigma^{-1}\mathbbm{1}^T
\end{pmatrix}^{-1}
\begin{pmatrix}
\mu_{0}\\
1
\end{pmatrix}.
\end{align*}

Let us observe that the matrix
\begin{align*}
\begin{pmatrix}
{\pmb\mu}\Sigma^{-1}{\pmb\mu}^T&{\pmb\mu}\Sigma^{-1}\mathbbm{1}^T\\
\mathbbm{1}\Sigma^{-1}{\pmb\mu}^T&\mathbbm{1}\Sigma^{-1}\mathbbm{1}^T
\end{pmatrix}
\end{align*}
is never singular as its determinant 
\begin{align*}
\left({\pmb\mu}\Sigma^{-1}{\pmb\mu}^T\right)\left(\mathbbm{1}\Sigma^{-1}\mathbbm{1}^T\right)-\left({\pmb\mu}\Sigma^{-1}\mathbbm{1}^T\right)^2>0,\,{\pmb \mu}\neq{\pmb 0}
\end{align*}
due to the Cauchy–Schwarz inequality.

Argumentation that the attained point is the minimum, is analogous to the one given in the proof of Proposition \ref{prop:var}.
\end{proof}

\begin{proof}[Proof of Proposition \ref{prop:EF_risk}.]
It is easy to compute that
\begin{align*}
\begin{pmatrix}
{\pmb\mu}\Sigma^{-1}{\pmb\mu}^T&{\pmb\mu}\Sigma^{-1}\mathbbm{1}^T\\
\mathbbm{1}\Sigma^{-1}{\pmb\mu}^T&\mathbbm{1}\Sigma^{-1}\mathbbm{1}^T
\end{pmatrix}^{-1}
=\frac{1}{d}
\begin{pmatrix}
\mathbbm{1}\Sigma^{-1}\mathbbm{1}^T&-{\pmb\mu}\Sigma^{-1}\mathbbm{1}^T\\
-{\pmb\mu}\Sigma^{-1}\mathbbm{1}^T&
{\pmb\mu}\Sigma^{-1}{\pmb\mu}^T
\end{pmatrix},
\end{align*}
where $d=\left({\pmb\mu}\Sigma^{-1}{\pmb\mu}^T\right)\left(\mathbbm{1}\Sigma^{-1}\mathbbm{1}^T\right)-\left({\pmb\mu}\Sigma^{-1}\mathbbm{1}^T\right)^2>0$ if ${\pmb \mu}\neq{\pmb 0}$. Then, by inserting $\mathbb{W}$ from Proposition \ref{prop:min_var_mu_const} into $\sigma^2_P=\mathbb{W}\Sigma\mathbb{W}^T$, we obtain
\begin{align*}
\sigma^2_P&=\frac{1}{d^2}
\left(\mu_0,\,1\right)
\begin{pmatrix}
\mathbbm{1}\Sigma^{-1}\mathbbm{1}^T&-{\pmb\mu}\Sigma^{-1}\mathbbm{1}^T\\
-{\pmb\mu}\Sigma^{-1}\mathbbm{1}^T&
{\pmb\mu}\Sigma^{-1}{\pmb\mu}^T
\end{pmatrix}
\begin{pmatrix}
\Sigma^{-1}{\pmb\mu}^T,\,
\Sigma^{-1}{\mathbbm{1}}^T
\end{pmatrix}^T
\left({\pmb\mu}^T,\,\mathbbm{1}^T\right)
\begin{pmatrix}
\mathbbm{1}\Sigma^{-1}\mathbbm{1}^T&-{\pmb\mu}\Sigma^{-1}\mathbbm{1}^T\\
-{\pmb\mu}\Sigma^{-1}\mathbbm{1}^T&
{\pmb\mu}\Sigma^{-1}{\pmb\mu}^T
\end{pmatrix}
\begin{pmatrix}
\mu_{0}\\
1
\end{pmatrix}\\
&=\frac{1}{d}\left(\mu_0,\,1\right)
\begin{pmatrix}
\mathbbm{1}\Sigma^{-1}\mathbbm{1}^T&-{\pmb\mu}\Sigma^{-1}\mathbbm{1}^T\\
-{\pmb\mu}\Sigma^{-1}\mathbbm{1}^T&
{\pmb\mu}\Sigma^{-1}{\pmb\mu}^T
\end{pmatrix}
\begin{pmatrix}
\mu_{0}\\
1
\end{pmatrix}\\
&=\frac{1}{d}\left(\left(\mathbbm{1}\Sigma^{-1}\mathbbm{1}^T\right)^2\mu_0^2-2\left({\pmb\mu}\Sigma^{-1}\mathbbm{1}^T\right)\mu_0+{\pmb\mu}\Sigma^{-1}{\pmb\mu}^T\right).
\end{align*}
\end{proof}

\begin{proof}[Proof of Proposition \ref{prop:EF_with_rf}.]
Let $P=w_1r_1+\ldots+w_nr_n+w_fr_f$ be the random variable, where $w_1+w_2+\ldots+w_n+w_f=1$, $(w_1,\,w_2,\,\ldots,\,w_n,\,w_f)\in\mathbb{R}^{n+1}$, the real-valued random variables $r_1,\,r_2,\,\ldots,\,r_n$ have positive variances and represent the returns "paid" by the financial firms $A_1,\,A_2,\,\ldots,\,A_n$; $r_f=const.$ denotes the risk-free return rate of asset $F$, and $w_f$ denotes the proportion invested to $F$. Then
\begin{align*}
\mu_P&=w_1\mu_{1}+w_2\mu_2+\ldots+w_n\mu_{n}+w_f\,r_f\\
&=\pmb{\mu}\mathbb{W}^T+w_f\,r_f\\
&=\pmb{\mu}\mathbb{W}^T+r_f(1-\mathbbm{1}\mathbb{W}^T)
\end{align*}
or
\begin{align*}
\mu_P-r_f=(\pmb{\mu}-r_f\mathbbm{1})\mathbb{W}^T
\Rightarrow
\tilde{\mu}_P=\tilde{\pmb{\mu}}\mathbb{W}^T,
\end{align*}
where $\tilde{\mu}_P=\mu_P-r_f$ and $\tilde{\pmb{\mu}}=\pmb{\mu}-r_f\mathbbm{1}$. It is evident that the variance of the random variable $P=w_1r_1+\ldots+w_nr_n+w_fr_f$ is $\sigma_P^2=\mathbb{W}\Sigma\mathbb{W}^T$, $\mathbb{W}\in\mathbb{R}^n$. Thus, we shall minimize the variance $\sigma_P^2=\mathbb{W}\Sigma\mathbb{W}^T$ under constraint $\tilde{\pmb{\mu}}\mathbb{W}^T=\mu_0-r_f=\tilde{\mu}_0=conts$. Let us set up the Lagrangian function
\begin{align*}
L(\mathbb{W},\,\lambda)=\mathbb{W}\Sigma\mathbb{W}^T-\lambda(\tilde{\pmb{\mu}}\mathbb{W}^T-\tilde{\mu}_0),
\end{align*}
whose partial derivatives with respect to $w_1,\,w_2,\,\ldots,\,w_n,\,w_f$ and $\lambda$ implies the system of linear equations
\begin{align*}
\begin{cases}
2\Sigma \mathbb{W}^T-\lambda\tilde{\pmb{\mu}}^T={\pmb 0}^T\\
\tilde{\pmb{\mu}}\mathbb{W}^T-\tilde{\mu}_0=0
\end{cases},
\end{align*}
and its solution is
\begin{align}\label{weights_with_rf}
\mathbb{W}^T=(\mu_0-r_f)\frac{\Sigma^{-1}\tilde{\pmb{\mu}}^T}{\tilde{\pmb{\mu}}\Sigma^{-1}\tilde{\pmb{\mu}}^T},\,\tilde{\pmb \mu}\neq{\pmb 0}.
\end{align}
Justification that the attained point is minimum, is implied due to the Hessian matrix of $L(\mathbb{W},\,\lambda)$ being $2\Sigma$.
Then
\begin{align*}
\sigma^2_P=\mathbb{W}\Sigma\mathbb{W}^T=(\mu_0-r_f)^2\frac{\tilde{\pmb{\mu}}^T\Sigma^{-1}}{\tilde{\pmb{\mu}}\Sigma^{-1}\tilde{\pmb{\mu}}^T}
\Sigma
\frac{\Sigma^{-1}\tilde{\pmb{\mu}}^T}{\tilde{\pmb{\mu}}\Sigma^{-1}\tilde{\pmb{\mu}}^T}
=\frac{(\mu_0-r_f)^2}{\tilde{\pmb{\mu}}\Sigma^{-1}\tilde{\pmb{\mu}}^T},
\,\tilde{\pmb \mu}\neq{\pmb 0}
\end{align*}
or
\begin{align*}
\mu-r_f=\pm\sqrt{\tilde{\pmb{\mu}}\Sigma^{-1}\tilde{\pmb{\mu}}^T}\,\sigma,\,\sigma\geqslant0.
\end{align*}
\end{proof}

\begin{proof}[Proof of Proposition \ref{prop:Tangent}.]
We shall search for the common points of the efficient frontiers given in Propositions \ref{prop:EF_risk} and \ref{prop:EF_with_rf}. In other words, we shall require 
$\mathbbm{1}\mathbb{W}^T=1$, in $w_1+w_2+\ldots+w_n+w_f=1$ which means that we set $w_f=0$. Then, if $\mathbbm{1}\mathbb{W}^T=1$, the equality \eqref{weights_with_rf} implies
\begin{align*}
(\mu_0-r_f)=\frac{\tilde{\pmb{\mu}}\Sigma^{-1}\tilde{\pmb{\mu}}^T}{\mathbbm{1}\Sigma^{-1}\tilde{\pmb{\mu}}^T},\,\tilde{\pmb \mu}\neq{\pmb 0}.
\end{align*}
By inserting this back to \eqref{weights_with_rf} we obtain
\begin{align*}
\mathbb{W}^T=\frac{\Sigma^{-1}\tilde{\pmb{\mu}}^T}{\mathbbm{1}\Sigma^{-1}\tilde{\pmb{\mu}}^T},\,\tilde{\pmb \mu}\neq{\pmb 0}
\end{align*}
and that is the optimal portfolio $M$, see Proposition \ref{prop:Sharpe}.
It remains to set up the line in $(\sigma,\,\mu)$ plane that intersects the points
$(0,\,r_f)$ and $(\sigma_M,\,\mu_M)$.
\end{proof}

\begin{proof}[Proof of Proposition \ref{prop:Separation}.]
Let $\mathbb{W}_1,\,\mathbb{W}_2,\,\ldots,\,\mathbb{W}_m$ denote portfolios that consist of differently weighted risky investments to the same financial firms $A_1$, $A_2$, $\ldots$, $A_n$. Let $\mu_{0,\,1},\,\mu_{0,\,2},\,\ldots,\,\mu_{0,\,m}$ be the returns of the portfolios $\mathbb{W}_1,\,\mathbb{W}_2,\,\ldots,\,\mathbb{W}_m$ respectively. Since $\mathbb{W}_1,\,\mathbb{W}_2,\,\ldots,\,\mathbb{W}_m$ are efficient, they all satisfy \eqref{weights_ef_no_rf} and $\mu_{0,\,i}\geqslant \mu_{\sigma \min}={\pmb \mu}\Sigma^{-1}\mathbbm{1}^T/\mathbbm{1}\Sigma^{-1}\mathbbm{1}^T$ for all $i=1,\,2,\,\ldots,\,m$, see \eqref{eq:EF_risky_only}. Then
\begin{align*}
&w_1\mathbb{W}_1^T+w_2\mathbb{W}_2^T+\ldots+w_m\mathbb{W}_m^T\\
&=\left(\Sigma^{-1}{\pmb \mu}^T,\,\Sigma^{-1}\mathbbm{1}^T\right)
\begin{pmatrix}
{\pmb\mu}\Sigma^{-1}{\pmb\mu}^T&{\pmb\mu}\Sigma^{-1}\mathbbm{1}^T\\
\mathbbm{1}\Sigma^{-1}{\pmb\mu}^T&\mathbbm{1}\Sigma^{-1}\mathbbm{1}^T
\end{pmatrix}^{-1}
\begin{pmatrix}
w_1\mu_{0,\,1}\\
w_1
\end{pmatrix}\\
&+
\left(\Sigma^{-1}{\pmb \mu}^T,\,\Sigma^{-1}\mathbbm{1}^T\right)
\begin{pmatrix}
{\pmb\mu}\Sigma^{-1}{\pmb\mu}^T&{\pmb\mu}\Sigma^{-1}\mathbbm{1}^T\\
\mathbbm{1}\Sigma^{-1}{\pmb\mu}^T&\mathbbm{1}\Sigma^{-1}\mathbbm{1}^T
\end{pmatrix}^{-1}
\begin{pmatrix}
w_2\mu_{0,\,2}\\
w_2
\end{pmatrix}\\
&\,\,\,\vdots\\
&+
\left(\Sigma^{-1}{\pmb \mu}^T,\,\Sigma^{-1}\mathbbm{1}^T\right)
\begin{pmatrix}
{\pmb\mu}\Sigma^{-1}{\pmb\mu}^T&{\pmb\mu}\Sigma^{-1}\mathbbm{1}^T\\
\mathbbm{1}\Sigma^{-1}{\pmb\mu}^T&\mathbbm{1}\Sigma^{-1}\mathbbm{1}^T
\end{pmatrix}^{-1}
\begin{pmatrix}
w_m\mu_{0,\,m}\\
w_m
\end{pmatrix}\\
&=\left(\Sigma^{-1}{\pmb \mu}^T,\,\Sigma^{-1}\mathbbm{1}^T\right)
\begin{pmatrix}
{\pmb\mu}\Sigma^{-1}{\pmb\mu}^T&{\pmb\mu}\Sigma^{-1}\mathbbm{1}^T\\
\mathbbm{1}\Sigma^{-1}{\pmb\mu}^T&\mathbbm{1}\Sigma^{-1}\mathbbm{1}^T
\end{pmatrix}^{-1}
\begin{pmatrix}
\tilde{\mu}_{0}\\
1
\end{pmatrix},
\end{align*}
where $\tilde{\mu}_0=w_1\mu_{0,\,1}+w_2\mu_{0,\,2}+\ldots+w_m\mu_{0,m}\geqslant \mu_{\sigma \min}$.

The proof is analogous because of \eqref{weights_with_rf} in case there is risky-less investment added next to financial firms $A_1,\,A_2,\,\ldots,\,A_n$ whose random returns $r_1,\,r_2,\,\ldots,\,r_n$ have positive variances. 

\end{proof}

\section{Example} In this section we choose some hypothetical data in $(\sigma,\,\mu)$ plane which represent the standard deviation and expectation of the random returns of certain investments and visualize the propositions from Section \ref{sec:Propop}. All the necessary computations and visualizations are performed using the program \cite{Matlab}.

Let us suppose that the random returns $r_1,\,r_2,\,\ldots,\,r_8$ of certain investments are described by their expectations and covariance matrix:
\begin{align*}
{\pmb \mu}&=( 0.0620,\,0.0660,\,0.0838,\,0.0849,\,0.0674,\,0.0949,\,0.6780,\,0.0691),\\
\Sigma&=
{\small\begin{pmatrix*}[r]
0.0099&0.0105&0.0124&0.0002&0.0023&0.0020&-0.0050&-0.0032 \\
0.0105&0.0114&0.0142&-0.0004&0.0024&0.0014&-0.0058&-0.0039  \\
0.0124&0.0142&0.0201&-0.0019&0.0026&0.0000&-0.0079&-0.0057 \\
0.0002&-0.0004&-0.0019&0.0207&0.0136&0.0216&0.0217&0.0197 \\
0.0023&0.0024&0.0026&0.0136&0.0353&0.0125&0.0109&0.0108 \\
0.0020&0.0014&0.0000&0.0216&0.0125&0.0239&0.0203&0.0188 \\
-0.0050&-0.0058&-0.0079&0.0217&0.0109&0.0203&0.0312&0.0262 \\
-0.0032&-0.0039&-0.0057&0.0197&0.0108&0.0188&0.0262&0.0229
\end{pmatrix*}}.
\end{align*}
and say that the risk-free investment rate $r_f=0.015$. Then, according to Proposition \ref{prop:var}, the absolute minimum variance portfolio (denoted by "$\sigma \min$" in subscription), its standard deviation and expectation are:
\begin{align*}
\mathbb{W}_{\sigma \min}&:=(0.4343,\,0.7324,\,-0.4033,\,0.5122,\,-0.0019,\,-0.6344,\,0.0394,\,0.3213),\\
(\sigma_{\sigma \min},\,\mu_{\sigma \min})&:=(0.0677,\,0.0495).
\end{align*}
According to Proposition \ref{prop:Sharpe}, the optimal portfolio, i.e. the portfolio of the maximal Sharpe ratio, (denoted by "$M$" in subscription), its standard deviation and expectation are:
\begin{align*}
\mathbb{W}_{M}&:=(1.2007,\,-1.5916,\,0.8996,\,0.5272,\,-0.0389,\,-0.0715,\,-0.1321,\,0.2066),\\
(\sigma_{M},\,\mu_{M})&:=(0.0966,\,0.0854).
\end{align*}
According to Proposition \ref{prop:EF_risk}, we get the efficient frontier curve in $(\sigma,\,\mu)$ plane
\begin{align}\label{EF_ex_no_F}
\sigma=\sqrt{3.7017\mu^2-0.3667\mu+0.0137},\,\mu\geqslant 0.0495
\end{align}
where 
\begin{align*}
\sigma=\sqrt{3.7017\mu^2-0.3667\mu+ 0.0137},\,\mu\in\mathbb{R}
\end{align*}
represents the minimum variance frontier.

According to Proposition \ref{prop:EF_with_rf}, the efficient frontier of the portfolio, including the risk-free investment, is
\begin{align}\label{EF_ex_with_F}
\mu= 0.7283\sigma+0.015,\,\sigma\geqslant0.
\end{align}
Because of Proposition \ref{prop:Tangent}, the semi-line \eqref{EF_ex_with_F} is also the tangent line of the efficient frontier \eqref{EF_ex_no_F} at the optimal portfolio $(\sigma_M,\,\mu_M)$. According to the proof of Proposition \ref{prop:EF_with_rf}, the semi-lines
\begin{align*}
\mu=\pm 0.7283\sigma+0.015,\,\sigma\geqslant0.
\end{align*}
represent the minimal variance frontier when the portfolio includes the risk-free investment. 

According to Proposition \ref{prop:min_var_mu_const}, we may choose the portfolios
\begin{align*}
&\mathbb{W}_1=(1.3834,\,-2.1453,\,1.2100,\,0.5308,\,-0.0477,\,0.0626,\,-0.1730,\,0.1792)\\
&\mathbb{W}_2=(1.9352,\,-3.8185,\,2.1480,\,0.5417,\,-0.0744,\,0.4679,\,-0.2964,\,0.0966)\\
&\mathbb{W}_3=(2,4379,\,-5.3426,\,3.0024,\,0.5515,\,-0.0987,\,0.8370,\,-0.4089,\,0.0213)
\end{align*}
whose standard deviations and expectations belong to the efficient frontier \eqref{EF_ex_no_F}. Then, by choosing $\tilde{\mu}_0=\mu_M=0.0854$ in system \eqref{syst}, we conclude that, for instance, the linear combination $0.8479\mathbb{W}_1+0.6809\mathbb{W}_2-0.5287\mathbb{W}_3$ also represents the portfolio lying on the efficient frontier: the optimal one in this particular case. If we choose $\tilde{\mu}_0=0.1319$, then the portfolio $0.0131\mathbb{W}_1+0.4545\mathbb{W}_2+0.5324\mathbb{W}_3$ is also efficient; its coordinates on $(\sigma,\,\mu)$ plane are $(0.1723,\,0.1319)$.

We depict all the listed thoughts of this section in the following Figure \ref{PicMarkowitz}.

\begin{figure}[H]
\centering
\caption{Visual of Markowitz's portfolio selection theory}
\includegraphics[width=160mm,height=130mm]{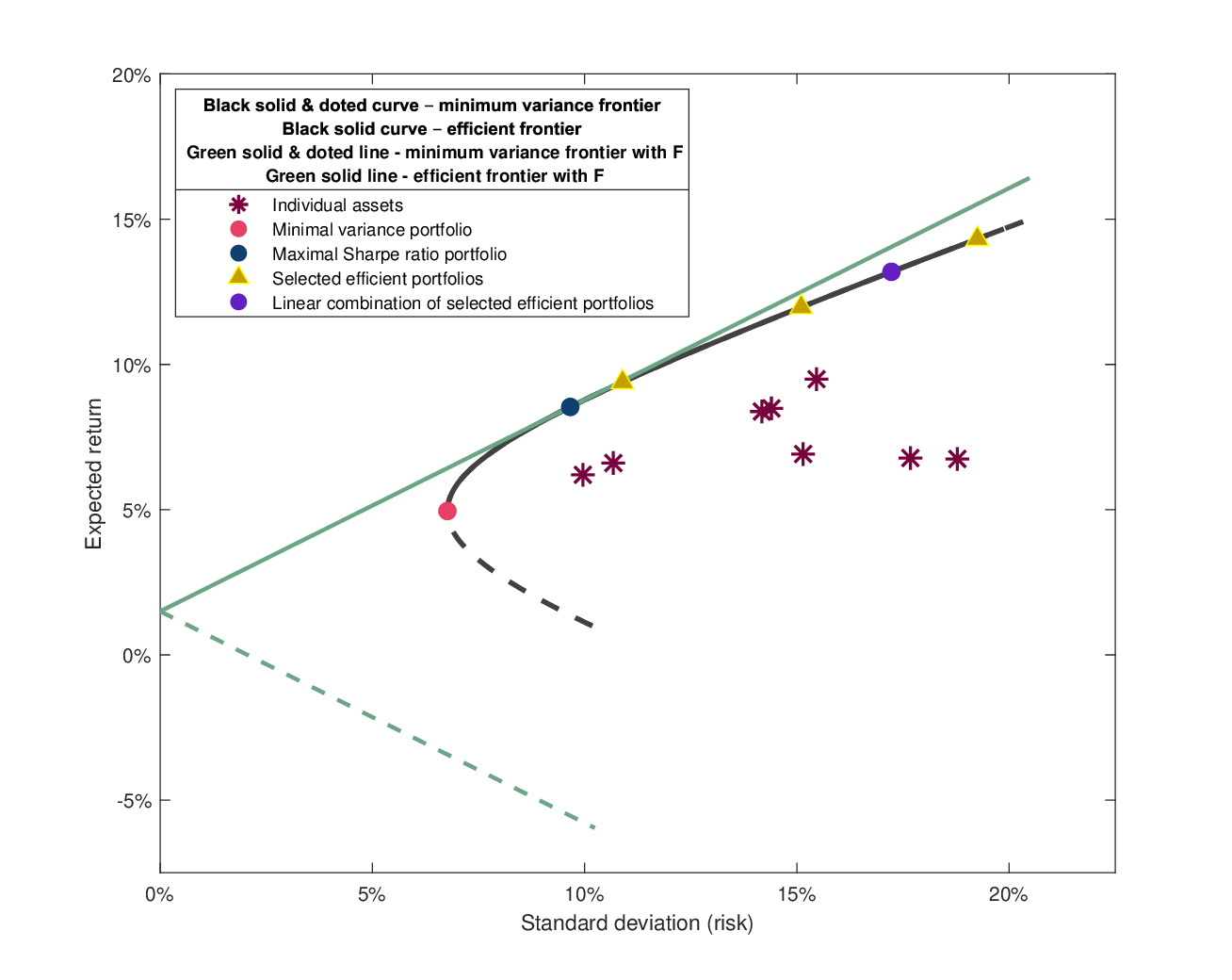}
\label{PicMarkowitz}
\end{figure}

\bibliography{references}
\end{document}